%% file: CDC_3.tex
\pdfoutput = 1
\documentclass[11pt, dvipsnames]{article}
\usepackage[margin=1in]{geometry}
\usepackage{setspace}
\usepackage[round,sort&compress]{natbib}
\let\cite\citep
\usepackage{amsmath}
\usepackage{amssymb}
\usepackage{algorithm}
\usepackage{varwidth}
\usepackage{booktabs}
\usepackage{algpseudocode}
\allowdisplaybreaks
\usepackage{hyperref}
\makeatletter
\newcommand{\printfnsymbol}[1]{%
  \textsuperscript{\@fnsymbol{#1}}%
}
\makeatother
\usepackage{graphicx}
\usepackage{caption}
\usepackage{natbib}
\usepackage{subcaption}
\usepackage{color}
\usepackage{bbm}
\usepackage{float}

\usepackage{amsthm}
\usepackage{dsfont}
\usepackage{mathrsfs}
\usepackage{lmodern}
\usepackage{hyperref}

\usepackage{mathtools}
\usepackage{verbatim}
\usepackage{subcaption}
\usepackage{enumitem}
\usepackage{amsthm}
\usepackage{macros}

\usepackage{booktabs}
\newtheorem{theorem}{Theorem}[section]
\newtheorem{lemma}[theorem]{Lemma}
\newtheorem{proposition}[theorem]{Proposition}

\newtheorem{remark}[theorem]{Remark}

\newtheorem{definition}[theorem]{Definition}

\usepackage{xcolor}

\usepackage{appendix}
\title{Incentive-Compatible Vertiport Reservation in Advanced \\ Air Mobility: An Auction-Based Approach
}

\author{Pan-Yang Su\thanks{Equal Contribution.} \thanks{EECS, University of California, Berkeley, CA 94720 (emails: \texttt{\{pan\_yang\_su, chinmay\_maheshwari, victoria\_tuck, sastry\} at berkeley dot edu}).} , Chinmay Maheshwari\footnotemark[1] \footnotemark[2] , Victoria Marie Tuck\footnotemark[2] , and Shankar Sastry\footnotemark[2]
}

\date{}

\begin{document}

\maketitle
\thispagestyle{empty}
\pagestyle{empty}

\begin{abstract}
\input{CDC_Sections/abstract2}
\end{abstract}

\newpage
\input{CDC_Sections/IntroductionVer3}

\input{CDC_Sections/ProblemSetup4}
\input{CDC_Sections/Results4}
\input{CDC_Sections/Conclusion2}

\bibliography{refs}
\bibliographystyle{plainnat}
\newpage
\appendix
\input{CDC_Sections/Appendix2}

\section{Table of Notations}\label{sec: TON}
\input{CDC_Sections/TableOfNotations}

\section*{ACKNOWLEDGMENTS}
Chinmay Maheshwari, Pan-Yang Su and Shankar Sastry acknowledge support from NSF Collaborative Research: Transferable, Hierarchical, Expressive, Optimal, Robust, Interpretable NETworks (THEORINET) Award No. DMS 2031899. Victoria Tuck and Shankar Sastry acknowledge support from Provably Correct Design of Adaptive Hybrid Neuro-Symbolic Cyber Physical Systems, Defense Advanced Research Projects Agency award number FA8750-23-C-0080. We thank Maria Gabriela Mendoza and Hamsa Balakrishnan for the helpful discussions. Microsoft Copilot was used to assist code development.
\clearpage


\addtolength{\textheight}{-12cm}   


\end{document}

%% file: CDC_Sections/abstract2.tex
The rise of advanced air mobility (AAM) is expected to become a multibillion-dollar industry in the near future. Market-based mechanisms are touted to be an integral part of AAM operations, which comprise heterogeneous operators with private valuations. In this work, we study the problem of designing a mechanism to coordinate the movement of electric vertical take-off and landing (eVTOL) aircraft, operated by multiple operators each having heterogeneous valuations associated with their fleet, between vertiports, while enforcing the arrival, departure, and parking constraints at vertiports. Particularly, we propose an incentive-compatible and individually rational vertiport reservation mechanism that maximizes a social welfare metric, which encapsulates the objective of maximizing the overall valuations of all operators while minimizing the congestion at vertiports. Additionally, we improve the computational tractability of designing the reservation mechanism by proposing a mixed binary linear programming approach that leverages the network flow structure.

%% file: CDC_Sections/IntroductionVer3.tex
\section{Introduction}

Advanced air mobility (AAM) encompasses the utilization of unmanned aerial vehicles (UAVs), air taxis, and various cargo and passenger transport solutions. This innovative approach taps into previously unexplored airspace, poised to revolutionize urban airspace. A recent report forecasts the air mobility market alone to exceed US\$50 billion by 2035, underlining this area's immense growth potential \cite{cohen2021urban}. 

Despite the widespread optimism surrounding AAM, the design of regulatory policies remains an open problem. While ideas from conventional air traffic management (e.g. \cite{doi:10.1287/opre.46.3.406, doi:10.1287/trsc.34.3.239.12300, doi:10.1287/opre.1100.0899, 4282854, 10.1007/978-3-642-86726-2_17}) could be leveraged, they often fall short in accommodating the dynamic and adaptable nature of AAM operations \cite{bichler2023airport}, resulting from on-demand requests from operators with heterogeneous private valuations
\cite{skorup2019auctioning, seuken2022market}.
Indeed, the administrative management methods prevalent in traditional air traffic management, such as grand-fathering rights, flow management, and first-come-first-serve, prove ineffective for AAM operations \cite{doi:10.2514/6.2020-2903, doi:10.2514/6.2020-2203} as these approaches fail to elicit the heterogeneous private valuations (arising from different aircraft specifications, demand realization, etc.) different operators have on using AAM resources. Furthermore, they risk fostering inefficient and anti-competitive outcomes, as evidenced in traditional airspace operations \cite{DIXIT2023102971}. Recognizing the need for tailored regulation, the Federal Aviation Administration (FAA) is actively developing a \emph{clean-slate} congestion management framework for AAM operations to ensure efficiency, fairness, and safety \cite{faa}.

Market-based congestion management mechanisms have been proposed as potential solutions for AAM operations \cite{chin2023traffic, balakrishnan2022cost, wang2023learning, doi:10.2514/6.2020-2203, skorup2019auctioning, seuken2022market}. 
\begin{figure}
\centering
\includegraphics[width=.8\linewidth]{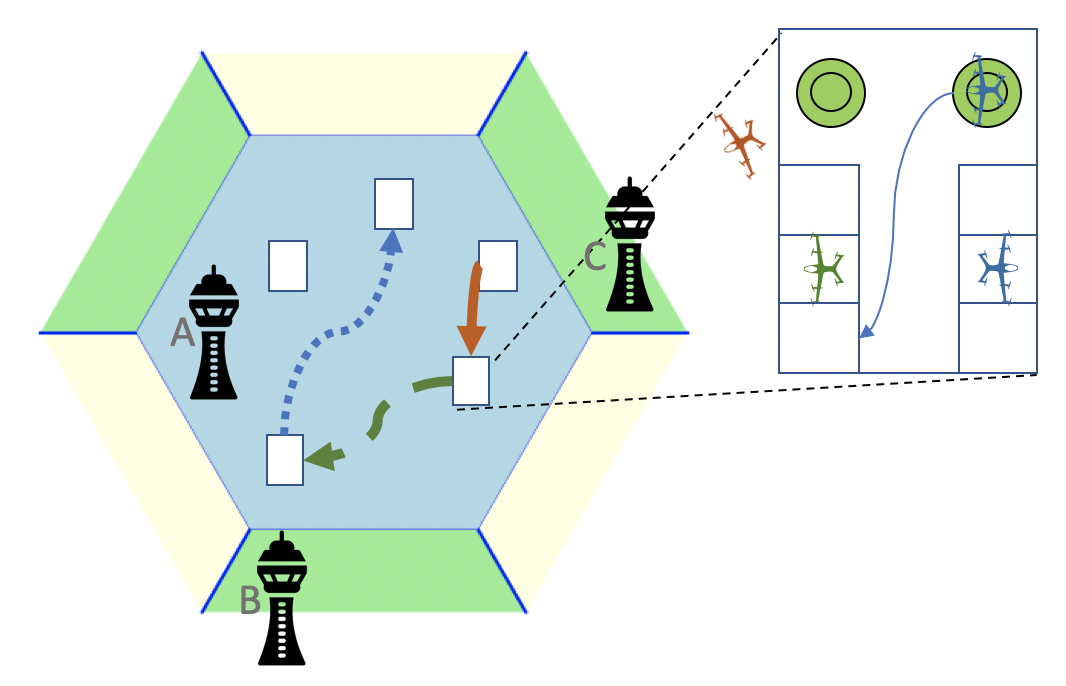}
\caption{Schematic representation of the air traffic network with a service provider tasked with coordinating the movement of aircraft of various fleet operators between vertiports in its domain. Each vertiport has a constraint on the number of arriving aircraft, departing aircraft, and parked aircraft.}
\label{fig: system model}
\end{figure}
 Even in conventional airspace management, market-based mechanisms are extensively studied such as \cite{BALL2018186, BASSO2010381, 6bec411b-7eef-3480-b725-7f8cb05b9529, MEHTA202041}, where both theoretical and empirical evidence show their precedence over administrative approaches \cite{DIXIT2023102971}. However, the design of market-based mechanisms that guarantee safety, efficiency, and fairness under the heterogeneous and on-demand nature of AAM operations has remained elusive as the existing approaches concentrate heavily on tactical deconfliction \cite{8569645, doi:10.2514/6.2020-0660}, while not accounting for efficiency, fairness and the economic incentives of operators \cite{9339882, sun2023fair, chin2023protocol, chin2023traffic, wang2023learning, balakrishnan2022cost, doi:10.2514/6.2020-2903, doi:10.2514/6.2020-2203}.
 
In this paper, we introduce an auction-based mechanism for a prominent AAM scenario of vertiport reservation, where electric vertical take-off and landing (eVTOL) operators with heterogeneous private valuations need to be coordinated to use vertiports based on their realized demands. This problem is challenging for three main reasons. First, the resulting reservation must ensure efficient, fair, and safe allocation of resources. Second, the operators may misreport their private valuations and demands to gain access to more valuable airspace resources (i.e. ensuring \emph{incentive compatibility}). Third, 
the computation of these auction mechanisms is combinatorial, as evidenced by existing air traffic flow management frameworks \cite{doi:10.1287/opre.46.3.406, doi:10.1287/trsc.34.3.239.12300, doi:10.1287/opre.1100.0899} (i.e. ensuring fast computability). 
Thus, the main question we set out for this work is: 
\begin{quote}
    \emph{
    How to design an efficient, fair, and safe vertiport reservation mechanism for heterogeneous and on-demand nature of eVTOL operators, while ensuring incentive compatibility and faster computation?
    }
\end{quote}

We consider an air transportation network (ATN) managed by a service provider (SP). The SP is responsible for ensuring the efficient, safe, and fair movement of aircraft operated by various fleet operators (FOs) between vertiports (as depicted in Figure \ref{fig: system model}). The goal of the SP is to maximize a metric of \emph{social welfare} that is comprised of two objectives: \((i)\) maximize the overall (weighted) valuations\footnote{We allow the SP to weigh FOs differently in order to encourage new-comers in this emerging market.} of all FOs, and \((ii)\) minimize excessive congestion at vertiports\footnote{{Note that we \emph{only} consider congestion at the vertiports in this work. An extension to airborne congestion is discussed in Subsection \ref{Sec: speed-up}.}}. Additionally, the SP must \((iii)\) enforce arrival, departure, and parking capacity constraints at vertiports, and \((iv)\) elicit truthful valuations from heterogeneous FOs in the form of bids. 

We propose an auction mechanism, to be used by the SP, that satisfies \((i)-(iv)\). In this mechanism, using the bids submitted by FOs, the SP allocates the resources by maximizing {social welfare}, subject to capacity constraints. 
Next, the SP charges each FO a payment based on the \emph{externality} imposed by them, which is assessed by the difference in the optimal social welfare of remaining FOs when this FO is included versus when it is excluded from the auction environment. Note that this payment mechanism is inspired by the generalized Vickrey–Clarke–Groves (VCG) mechanism \cite{Nisan_Roughgarden_Tardos_Vazirani_2007}. We theoretically study the properties of the proposed mechanism in terms of incentive compatibility, individual rationality, and social welfare maximization (cf. Theorem \ref{thm: multi-capacity proof}). 

There are two computational challenges associated with designing this mechanism. First, naively optimizing social welfare over the set of feasible allocations could be computationally challenging. Therefore, we frame the problem as a mixed binary linear program by constructing a network-flow graph to reduce the number of binary variables. Second, the computation of externality in the payment mechanism, which requires maximizing social welfare over the set of feasible allocations, requires characterizing the set of feasible allocations when an FO is excluded from the auction environment, which is non-trivial as the underlying resource allocation problem is an exchange problem. Therefore, we introduce the idea of \emph{pseudo-bids}, where we simply set a bid of \(0\) to an FO while computing the optimal allocation when this FO is excluded from the auction environment. 

We note two important features of the problem we study in this work. First, {we focus only on strategic deconfliction where the safety is encoded in the form of minimizing congestion and ensuring capacity constraints, and not on tactical deconfliction. However, our approach can be integrated into the airborne automation workflow proposed in \cite{wei2023arrival} to also account for tactical deconfliction.
 } 
 Second, this problem is an ``exchange problem", where some of the resources desired by any FO could be occupied by aircraft of other FOs, and a feasible allocation in this setting needs to exchange the resources between FOs while respecting capacity constraints. In constrast, the standard slot allocation problems studied in conventional air traffic literature (cf. \cite{DIXIT2023102971, MEHTA202041, BALL2018186, 68d5566f-a8c4-3aac-89e0-3dad8b73de7f, PERTUISET201466, doi:10.1287/trsc.2020.0995, bichler2023airport}) are ``assignment problems'' where the slots need to be assigned to airlines and not exchanged between airlines.

\emph{Notation: }We denote the set of real numbers by \(\mathbb{R}\), non-negative real numbers by \(\mathbb{R}_+\), integers by \(\mathbb{Z}\), non-negative integers by \(\mathbb{Z}_+\), and natural numbers by \(\mathbb{N}\). For $N \in \mathbb{N}$, we define $[N] := \{1, 2, ..., N\}$. The indicator function is denoted as $\mathbf{1}(\cdot)$, which is 1 when $(\cdot)$ is true and 0 otherwise.
When indexing a set $b  = \{b_1, b_2, ..., b_N\}$, we follow the standard game-theoretic notation: $b_{-i} := \{b_1, ..., b_{i-1}, b_{i+1}, ..., b_N\}$.

%% file: CDC_Sections/ProblemSetup4.tex
\section{Problem Setup}
\subsection{System Model}
\label{system model}
We consider an air transportation network (ATN), comprised of multiple vertiports, which are used by electric vertical take-off and landing (eVTOL) aircraft. 
We focus on a strategic deconfliction mechanism that complements the tactical deconfliction algorithms proposed in \cite{8569645, 9756871,  doi:10.2514/6.2020-0660, SHAO2021103385}. 
The scheduling mechanism proceeds over non-overlapping time slots with a receding time horizon. 
At the beginning of each time slot, all fleet operators (FOs) submit a \emph{menu} of desired origin-destination pairs and the corresponding bids specifying how much they are willing to pay for getting scheduled.
Then, the service provider (SP) will compute a feasible allocation and payment and execute them in the next time slot.
The granted aircraft can now go to their desired locations. In most congested vertiports, when the parking capacity is fully utilized, any additional arrival would necessitate a simultaneous departure of an aicraft from that vertiport.
Thus, this is an ``exchange problem'' as opposed to the ``assignment problem'' studied in other air traffic allocation problems \cite{DIXIT2023102971, MEHTA202041, BALL2018186, 68d5566f-a8c4-3aac-89e0-3dad8b73de7f, PERTUISET201466, doi:10.1287/trsc.2020.0995, bichler2023airport}.

We denote the set of vertiports by $\sector$, the set of FOs by $\OptSet$, and the set of eVTOL aircraft by $\UAVSet$. We consider the problem for \(H\) time slots. 

\paragraph{Vertiports} 
At any time $t \in [\horizon]$, each vertiport $r \in \sector$ has three kinds of capacity constraints \footnote{{The arrival, departure and parking capacity constraints in our model are exogeneously determined at every time step and are un-correlated between two consecutive time steps. Extending our model to account for correlations is an interesting direction of future research.}}: \textit{(i)} \emph{arrival capacity constraints}, denoted by $\ArrivalCap(r, t) \in \mathbb{Z}_+$, that restrict the number of eVTOLs that can land at vertiport \(r\) at time \(t\); \textit{(ii)} \emph{departure capacity constraints}, denoted by $\DepartureCap(r, t) \in \mathbb{Z}_+$, that restrict the number of eVTOLs that can depart from vertiport \(r\) at time \(t\); \textit{(iii)} \emph{parking capacity constraints}, denoted by $\capacity(r, t) \in \mathbb{Z}_+$, that restrict the number of eVTOLs that can park at vertiport \(r\) at time \(t\). 
\paragraph{Fleet Operators}
Let $\OptsUAVNum_{\Opt}$ be the fleet of aircraft operated by FO $\Opt \in \OptSet$, and $\UAVSet := \{\UAV_{i,j}| i \in \OptSet, j \in \OptsUAVNum_{\Opt}\}$ be the set of all aircraft using the ATN. Each aircraft $\UAV_{i, j}$ is identified by a tuple $\left( \origsector_{i, j}, \UAVMenuNum_{i, j}, \{ \depart_{i, j, k}, \arrive_{i, j, k}, \val_{i, j, k}, \bid_{i, j, k}, \finalsector_{i, j, k}\}_{k \in \UAVMenuNum_{i, j}} \right)$ where 
\((i)\) \(\origsector_{i, j}\in \sector\) is the origin vertiport of aircraft \(\UAV_{i,j}\), \((ii)\) \(\UAVMenuNum_{i, j}\) is the menu of available routes to aircraft \(\UAV_{i,j}\); \((iii)\) any route \(k\in \UAVMenuNum_{i, j}\) implies that aircraft \(\UAV_{i,j}\) departs from \(\origsector_{i, j}\in\sector\) at time $\depart_{i, j, k}$ to arrive at $\finalsector_{i, j, k}\in\sector$ at time $\arrive_{i, j, k}$; \((iv)\) \(\val_{i, j, k}\) denotes the private valuation of aircraft \(\UAV_{i,j}\) to choose the route \(k\in \UAVMenuNum_{i,j}\); and \((v)\) $\bid_{i, j, k}\in \mathbb{R}_+$ is the bid submitted by FO \(i\) to schedule aircraft \(\UAV_{i,j}\) on route \(k\in\UAVMenuNum_{i,j}\). {Note that we include the option to stay parked at the same vertiport in $\UAVMenuNum_{i,j}$, denoted by \(\varnothing\), and set its departure time to \(0\).}

{
Additionally, we denote the joint bid profile of all aircraft operated by FO \(i\in\OptSet\) by $\OptBid_i := (\bid_{i, j, k})_{j \in \OptsUAVNum_i, k\in \UAVMenuNum_{i,j}}$ and joint valuation profile of its fleet by $\OptVal_i := (\val_{i, j, k})_{j \in \OptsUAVNum_i, k\in \UAVMenuNum_{i,j}}$. For succinct notation, we denote the joint bid and valuation profile of all FOs as $\bidSet := (\OptBid_i)_{i \in \OptSet}$ and $\valSet := (\OptVal_i)_{i\in \OptSet}$, respectively.}

\subsection{Problem Formulation}
We consider an SP tasked with coordinating\footnote{We do 
not impose the information 
sharing constraints in 
\cite{balakrishnan2022cost, chin2023protocol}, where 
different sectors have different operators, and an SP only provides the identities,  but not the positions, of  aircraft to neighboring  sectors. We follow the  architecture in the current 
ATFM framework \cite{10.1007/978-3-642-86726-2_17, doi:10.1287/opre.46.3.406, doi:10.1287/trsc.34.3.239.12300, doi:10.1287/opre.1100.0899, 4282854}, where a central SP can aggregate 
information from all the sectors and make decisions.} the movement of aircraft by allocating them to their desired vertiports while ensuring that the capacity constraints are met. Formally, the SP needs to decide on a feasible allocation $\allocation = (\allocation_{i, j, k} \in \{0,1\}| i \in \OptSet, j \in \OptsUAVNum_i, k \in \UAVMenuNum_{i, j}) 
$, where
\begin{align*}
    \allocation_{i, j, k} = \begin{cases}
        1, & \text{if aircraft \(\UAV_{i,j}\) is allocated route \(k\in\UAVMenuNum_{i,j}\),} \\ 
        0, & \text{otherwise}.
    \end{cases}
\end{align*}

Given an allocation $\allocation$, let $\statesector(r, t, \allocation) \in \mathbb{Z}_+$ denote the number of aircraft occupying the parking spots at vertiport \(r\in \sector\) at time $t \in [\horizon]$. 
For every \(r\in\sector\), the initial occupation \(\statesector(r, 1, \allocation)\) is
\begin{equation*}
    S(r,1,x) = \sum_{i\in\OptSet}\sum_{j\in \OptsUAVNum_i} \mathbf{1}(\origsector_{i, j} = r).
\end{equation*}
For concise notation, we shall denote \(S(r,1,x)\) by \(\bar{S}(r)\) for every \(r\in \sector\) since it does not depend on $x$. 
Naturally, it must hold that, for every \(r\in \sector, t\in \{2,\ldots,\horizon\}\),
\begin{equation}
\label{eq: StateEvolution}
\begin{aligned}
    S(r,t,x) &= S(r,t-1,x) + \sum_{i \in \OptSet } \sum_{j \in \OptsUAVNum_i} \sum_{k \in \UAVMenuNum_{i, j}} \allocation_{i, j, k} \mathbf{1}(\finalsector_{i, j, k} = r, \arrive_{i, j, k} = t) \\&\quad  - \sum_{i\in\OptSet}\sum_{j\in \OptsUAVNum_i} \sum_{k \in \UAVMenuNum_{i, j}} \allocation_{i, j, k} \mathbf{1}(\origsector_{i, j} = r, \depart_{i, j, k} = t),
\end{aligned}
\end{equation}
where the second (resp. third) term on the RHS in the above equation denotes the set of incoming (resp. departing) aircraft in vertiport \(r\) at time \(t\). The \emph{residual capacity} at vertiport \(r\in \sector\) at time \(t\in [H]\) is $\residualsector(r, t, \allocation) := \capacity(r, t) - \statesector(r, t, \allocation)$. To ensure the existence of a feasible allocation as defined later in \eqref{eq: feasible}, we assume that $\capacity(r, t) - \bar{S}(r) \geq 0, \forall r\in \sector, t \in [H]$.

An allocation \(x\) is called feasible if it satisfies the following constraints:
\begin{itemize}
    \itemindent=4pt
    \item[(C1)] Each aircraft is allocated at most one route. That is, for every \(i \in \OptSet, j \in \OptsUAVNum_i\), $\sum_{k \in \UAVMenuNum_{i, j}}\allocation_{i, j, k} \leq 1$.
    \item[(C2)] Arrival and departure capacity constraints must be satisfied at every vertiport $r$ at all times. That is, for every \(r\in \sector, t\in [\horizon]\),
    \begin{align*}
        \sum_{i \in \OptSet } \sum_{j \in \OptsUAVNum_i} \sum_{k \in \UAVMenuNum_{i, j}} \allocation_{i, j, k} \mathbf{1}(\finalsector_{i, j, k} = r, \arrive_{i, j, k} = t) \leq \ArrivalCap(r, t), \\ 
        \sum_{i\in\OptSet}\sum_{j\in \OptsUAVNum_i} \sum_{k \in \UAVMenuNum_{i, j}} \allocation_{i, j, k} \mathbf{1}(\origsector_{i, j} = r, \depart_{i, j, k} = t) \leq \DepartureCap(r, t).
    \end{align*}
    \item[(C3)] Parking capacity constraints must be satisfied. That is, for every vertiport $r\in \sector$ at any time $t\in [\horizon]$, $\residualsector(r, t, x) \geq 0$.
\end{itemize}

Consequently, we define 
\begin{align}\label{eq: feasible}
\allocationSet := \left\{\allocation \in \{0, 1\}^{\sum_{i \in \OptSet} \sum_{j \in \OptsUAVNum_i} |\UAVMenuNum_{i, j}|}\Big|~\text{$\allocation$ satisfies (C1)-(C3)}\right\}
\end{align}
to be the set of feasible allocations. 


\begin{definition}[Social Welfare]\label{def: SW}
    Given $\allocation\in \allocationSet$, \emph{social welfare} is defined as follows.
\begin{align}\label{eq: SocWelfare}
\socialWelfare(\allocation; \valSet)\!\!:= \!\!\sum\limits_{i\in \OptSet} \rcof_i \sum\limits_{j \in \OptsUAVNum_i} \sum\limits_{k \in \UAVMenuNum_{i, j}}\!\!\!\!\val_{i, j, k} \cdot \allocation_{i, j, k} - \cc\!\!\sum\limits_{r \in \sector} \sum\limits_{t \in [\horizon]} \!\!\!\!\cpc_{r, t}(\statesector(r, t, x)),
\end{align}
where \((i)\) $\rcof_i \in \mathbb{R}_+$ is the weight factor specifying the relative importance of different FOs\footnote{Similar weight factors, termed as \emph{remote city opportunity factor}, are used in \cite{DIXIT2023102971}.}, \((ii)\) $\cpc_{r, t}: \mathbb{Z}_+ \rightarrow \mathbb{R}_+$ with $\cpc_{r, t}(0) = 0$ is discrete convex\footnote{Based on \cite{dca}, a function $f: \mathbb{Z} \rightarrow \mathbb{R}$ is discrete convex if $f(x+1) - f(x) \geq f(x) - f(x-1), \ \forall x \in \mathbb{Z}$.} to capture increasing marginal cost of congestion\footnote{{While we only consider the congestion resulting from parked aircraft, it is straightforward to extend our formulation to arriving and departing aircraft; see Subsection \ref{Sec: speed-up}.}}, and $(iii)~ \cc \in \mathbb{R}_+$ is the ratio between the congestion cost and the cumulative weighted valuations of FOs.
Furthermore, we define an optimal allocation as
\begin{align}
\optimalAllocation(\valSet) \in \underset{\allocation\in \allocationSet}{\arg\max} \ \socialWelfare(\allocation; \valSet),    
\end{align}
where ties are resolved arbitrarily. 
\end{definition}

\begin{remark}
    The social welfare objective \eqref{eq: SocWelfare} captures three main desiderata: efficiency, fairness, and safety. The objective \eqref{eq: SocWelfare} incorporates efficiency through additive valuations of FOs. Additionally, it incorporates the \emph{proportional fairness } criterion\footnote{We emphasize that {the fairness is at the FO-level.}} by assigning different weights to the valuations of different FOs, denoted by \((\rho_i)_{i\in\OptSet}\). Well-constructed weights can prevent larger FOs from monopolizing the resources; {for example, using the logarithm of the number of aircraft as an FO's weight.} Finally, it encompasses safety considerations in two ways: first, through capacity constraints; and second, by introducing a congestion-dependent term in \eqref{eq: SocWelfare} that penalizes vertiports when the number of aircraft increases. With these three considerations, the definition of social welfare aligns closely with that presented in \cite{DIXIT2023102971}.
\end{remark}

We assume the SP does not have access to the true valuations $\valSet$, as it is private information. Instead, the SP must use bids \(\bidSet\) reported by the FOs to allocate the aircraft to vertiports through an auction mechanism. More formally, given a bid profile
$\bidSet$, the SP uses a mechanism $\mechanism =(\allocationMech, (\pricingMech_i)_{i\in\OptSet})$, where for a given bid profile \(\bidSet\), 
 \((i)\) \(\allocationMech(\bidSet)\in \allocationSet\) is the allocation proposed by the mechanism; and \((ii)\) \(\pricingMech_i(\bidSet) \in \mathbb{R}\) denotes the payment charged to FO \(i\in \OptSet\). 
 Under the mechanism \(\mechanism\), the utility derived by any FO \(i\in \OptSet\) is
\begin{align}\label{eq: UtilityUAV}
    \utilityUAV_i(\bidSet;\mechanism) = \sum_{j \in \OptsUAVNum_i} \sum_{k \in \UAVMenuNum_{i, j}}\val_{i, j, k} \textbf{1} (\allocationMech_{i, j, k}(\bidSet)) - \pricingMech_i(\bidSet).
\end{align}
Given any arbitrary valuation profile \(V\), the goal is to design a vertiport reservation mechanism $\mechanism = (\allocationMech, \pricingMech)$ with the following desiderata.
\begin{itemize}
    \itemindent=4pt
    \item[\hypertarget{D1}{(D1)}] \textbf{Incentive Compatibility (IC):} Bidding truthfully is each FO's (weakly) dominant strategy, i.e., for every \(i\in \OptSet\), \( \OptBid_{-i} \in \mathbb{R}_+^{\sum_{\ell\in\OptSet\backslash\{i\}}\sum_{j\in\OptsUAVNum_{\ell }}|\UAVMenuNum_{\ell,j}|}\),
    \begin{align*}
    \OptVal_i \in \underset{\OptBid_i\in \mathbb{R}_+^{\sum_{j\in\OptsUAVNum_i}|\UAVMenuNum_{i,j}|}}{\arg\max}~ \utilityUAV_i(\OptBid_i,\OptBid_{-i};\mechanism).
    \end{align*}
    \item[\hypertarget{D2}{(D2)}] \textbf{Individual Rationality (IR):}  Bidding truthfully results in non-negative utility, i.e., for every \(i\in \OptSet\),
    \begin{equation*}
    \utilityUAV_i(\OptVal_i, \OptBid_{-i}; \mechanism) \geq 0, \quad \forall \ \OptBid_{-i} \in \mathbb{R}_+^{\sum_{\ell\in\OptSet\backslash\{i\}}\sum_{j\in\OptsUAVNum_{\ell }}|\UAVMenuNum_{\ell,j}|}.
    \end{equation*}
    \item[\hypertarget{D3}{(D3)}] {\textbf{Social Welfare Maximization (SWM):} The resulting allocation maximizes social welfare, i.e., 
    \begin{align*}
    \allocationMech(B) \in \underset{\allocation\in \allocationSet}{\arg\max}
        ~\socialWelfare(\allocation;\valSet).
    \end{align*}}
\end{itemize}

%% file: CDC_Sections/Results4.tex
\section{Mechanism Design}
In this section, we present an auction mechanism that satisfies \hyperlink{D1}{(D1)}-\hyperlink{D3}{(D3)} in Subsection \ref{ssec: Mechanism} and prove its theoretical properties in Subsection \ref{ssec: TheoreticalAuction}. We defer the optimization algorithm to Section \ref{ssec: Computational}.

\subsection{Mechanism}
\label{ssec: Mechanism}
Inspired by Myerson's lemma \cite{doi:10.1287/moor.6.1.58},  our approach is to separate the allocation and payment functions so that the latter can ensure IC and IR as long as the former ensures maximization of total welfare in terms of bids submitted.

\noindent \textbf{Allocation Function:} Given a bid profile \(\bidSet\in \mathbb{R}^{\sum_{i\in\OptSet}\sum_{j\in\OptsUAVNum_i}|\UAVMenuNum_{i,j}|}_+\), the allocation is obtained by 
\begin{align}\label{eq: AllocationMultipleCapacity}
\multiCapAllocation(\bidSet) \in \underset{\allocation\in \allocationSet}{\arg\max} ~\SW(\allocation;\bidSet).
\end{align} 

\noindent \textbf{Payment Function:}
We first define a function \(\pseudoBid: \OptSet \times  \mathbb{R}_+^{{\sum_{i\in\OptSet}\sum_{j\in\OptsUAVNum_i}|\UAVMenuNum_{i,j}|}} \rightarrow \mathbb{R}_+^{{\sum_{i\in\OptSet}\sum_{j\in\OptsUAVNum_i}|\UAVMenuNum_{i,j}|}}\) such that for any \(\ell\in \OptSet\) and bid \(\bidSet\in \mathbb{R}_+^{{\sum_{i\in\OptSet}\sum_{j\in\OptsUAVNum_i}|\UAVMenuNum_{i,j}|}}\), 
\begin{align}\label{eq: pseudoBid}
    \pseudoBid_{i,j,k}(\ell, \bidSet) = \begin{cases}
        b_{i,j,k}, & \text{if} \ i \neq \ell, \\ 
        0, & \text{if} \ i = \ell,
    \end{cases} \ \forall \ i\in \OptSet, j\in \OptsUAVNum_i, k\in \UAVMenuNum_{i,j}.
\end{align}
The payment function, given a bid profile \(\bidSet\), is
\begin{equation}
\label{eq: paymentMC}
\multiCapPricing_i(\bidSet) \!\!= \!\! \frac{1}{\rcof_i}\left(\max_{x'\in \allocationSet}\!\totalBid_{-i}(x'; \pseudoBid(i,\bidSet)) \!- \!\totalBid_{-i}(\multiCapAllocation; \bidSet)\right),
\end{equation}
where for every \(i\in \OptSet\), $\allocation\in \allocationSet$, and $\bidSet\in \mathbb{R}_+^{{\sum_{i\in\OptSet}\sum_{j\in\OptsUAVNum_i}|\UAVMenuNum_{i,j}|}}$, 
\begin{equation}
\label{eq: RemainingTB}
\totalBid_{-i}(\allocation;\bidSet) := \!\!\!\!
\sum\limits_{\ell\in \OptSet_{-i}} \!\!\!\!\rcof_\ell \sum\limits_{j \in \OptsUAVNum_\ell} \sum\limits_{k \in \UAVMenuNum_{\ell, j}}\!\!\!\!\bid_{\ell, j, k} \cdot \allocation_{\ell, j, k} - \cc\sum\limits_{r \in \sector} \sum\limits_{t \in [\horizon]}\cpc_{r, t}(\statesector(r, t, \allocation)).
\end{equation}

\begin{remark}
The payment rule is inspired by the VCG mechanism, where each FO is charged a payment based on the \emph{externality} created by them. Particularly, the typical VCG payment for any player is determined by assessing the difference in the optimal social welfare of players when they are present, versus when they are excluded from the auction environment. 
\end{remark}

\begin{remark}
There are some notable differences between the VCG payment and \eqref{eq: paymentMC}. First, since our problem is an ``exchange problem'' and not the typical ``assignment problem'', we need to be cognizant of the physical resources occupied by the aircraft of that operator. However, this would require us to enumerate all the feasible combinations if we were to directly implement VCG mechanisms. To overcome the problem of enumerating all feasible solutions while computing payments, we adopt a novel approach of ``pseudo-bids'', where while computing the payments, each non-participating aircraft is considered to be using a bid of \(0\), as formally described in \eqref{eq: pseudoBid}.

Second, since the objective function \eqref{eq: SocWelfare} is not the summation of the participants' valuations, the typical VCG auction is not directly applicable. Instead, we follow \cite{DIXIT2023102971, Nisan_Roughgarden_Tardos_Vazirani_2007} to devise the payment rule for any $i\in\OptSet$ and $\bid\in \mathbb{R}_+^{\sum_{j\in\OptsUAVNum_i}|\UAVMenuNum_{i,j}|}$.
\end{remark}

\subsection{Theoretical Analysis}\label{ssec: TheoreticalAuction}
\begin{theorem}
\label{thm: multi-capacity proof}
The proposed mechanism $\multiCapMech := (\multiCapAllocation, \multiCapPricing)$, {defined by (\ref{eq: AllocationMultipleCapacity}) and (\ref{eq: paymentMC})} is IC, IR, and SWM.
\end{theorem}

\begin{proof}
Observe from \eqref{eq: SocWelfare} that $\SW(\allocation;\valSet)$ is a weighted summation of FOs' valuations and the congestion cost.
Since the congestion cost is independent of valuations, $\bar{x}(V) \in \arg\max_{x\in X}\SW(x;\valSet)$ is an affine maximizer with respect to FOs' valuations, as defined in \cite[Definition 9.30]{Nisan_Roughgarden_Tardos_Vazirani_2007}. Thus, the allocation function (\ref{eq: AllocationMultipleCapacity}) and the payment function (\ref{eq: paymentMC}) form a generalized VCG mechanism, and IC directly follows from \cite[Proposition 9.31]{Nisan_Roughgarden_Tardos_Vazirani_2007}. 
Finally, IR follows from \cite[Lemma 9.20]{Nisan_Roughgarden_Tardos_Vazirani_2007} since the bids are non-negative, and the allocation is an affine maximizer, as formally proved below.  

For any \(B_{-i}\in \mathbb{R}_+^{\sum_{\ell\in \OptSet\backslash\{i\}, j\in\OptsUAVNum_\ell}|\UAVMenuNum_{\ell, j}|}\),  
\begin{align}\label{eq: IR}
&\utilityUAV_i(V_i,\bidSet_{-i};\mechanism) \!=\!\! \sum_{j \in [\OptsUAVNum_i]} \sum_{k \in [\UAVMenuNum_{i, j}]}\!\!\!\!\!\!\val_{i, j, k} \textbf{1} (\allocationMech_{i, j, k}(V_i, \bidSet_{-i})) - \pricingMech_i(V_i, \bidSet_{-i})\notag \\
&= \frac{1}{\rho_i}\bigg(\rho_i\sum_{j \in [\OptsUAVNum_i]} \sum_{k \in [\UAVMenuNum_{i, j}]}\!\!\!\!\val_{i, j, k} \textbf{1} (\allocationMech_{i, j, k}(V_i, \bidSet_{-i}))+ \!\totalBid_{-i}(\multiCapAllocation; V_i, \bidSet_{-i}) \! \notag \\ &\hspace{3cm}-\max_{x'\in \allocationSet}\!\totalBid_{-i}(x'; \pseudoBid(i,V_i, \bidSet_{-i}))\bigg) \notag \\ 
&= \frac{1}{\rho_i}\bigg(\totalBid(\multiCapAllocation; V_i, \bidSet_{-i}) -\max_{x'\in \allocationSet}\!\totalBid_{-i}(x'; \pseudoBid(i,V_i, \bidSet_{-i}))\bigg) .
\end{align}

Since \(\bar{x}\in \arg\max_{x\in X} \totalBid(\multiCapAllocation; V_i, \bidSet_{-i})\), it holds that \(\totalBid(\multiCapAllocation; V_i, \bidSet_{-i}) \geq \totalBid(x^\dagger; V_i, \bidSet_{-i}),\) where \(x^\dagger \in \arg\max_{x'\in \allocationSet}\!\totalBid_{-i}(x'; \pseudoBid(i,V_i, \bidSet_{-i}))\). Thus, we obtain 
\begin{align*}
    &\utilityUAV_i(V_i,\bidSet_{-i};\mechanism) \!\geq\! \frac{1}{\rho_i}\!\bigg(\!\!\totalBid(x^\dagger; V_i, \bidSet_{-i}) \!-\!\totalBid_{-i}(x^\dagger; \pseudoBid(i,V_i, \bidSet_{-i}))\!\!\!\bigg) \\
    &= \frac{1}{\rho_i}\sum_{j \in [\OptsUAVNum_i]} \sum_{k \in [\UAVMenuNum_{i, j}]}\val_{i, j, k} \textbf{1} (x^\dagger_{i, j, k}(V_i, \bidSet_{-i})) \geq 0. 
\end{align*}

\end{proof}

\begin{figure*}[!ht]
        \centering
    \begin{minipage}[l]{0.82\textwidth}
        \centering
    \includegraphics[width=\textwidth]{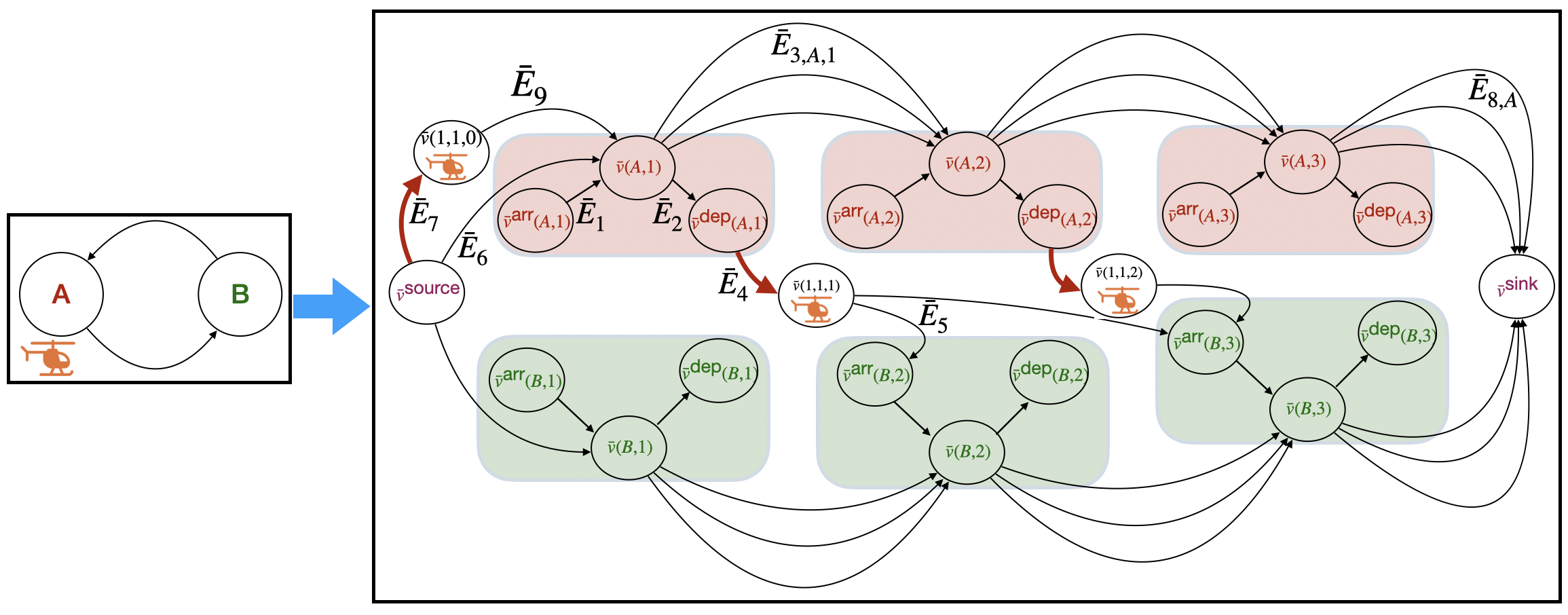}
    \end{minipage}
      \caption{Auxiliary graph \(\bar{G}\) constructed from an ATN with two vertiports and one aircraft over three time slots.
      } 
      \label{fig: AuxGraph}
\end{figure*}

\section{Optimization Algorithm}\label{ssec: Computational}
In this section, we formulate \eqref{eq: AllocationMultipleCapacity} as a mixed binary linear program (MBLP), as shown in \eqref{eq: ILPNew}. We derive this in three steps. First, in Subsection \ref{sec: Auxiliary Graph Construction}, we construct a time-extended flow network, where vertices are vertiport-time and aircraft-time pairs with edges capturing capacity constraints and route allocation. Then, using binary variables ($\delta_{i,j,\tau}$ as formally defined later in \eqref{eq: ILP_BinaryVariable} and \eqref{eq: ILP_VariableDelta}) to ensure that each aircraft is allocated one route, we formulate a mixed integer linear program (MILP \eqref{eq: ILP}) in Subsection \ref{ssec: MILP}. This MILP has fewer binary variables than \eqref{eq: AllocationMultipleCapacity} when the number of unique departure times for any aircraft is less than the size of its menu. Finally, in Subsection \ref{ssec: MBLP}, we show that the total unimodularity of the constraint matrix ($\incidence_\star$ in \eqref{eq: ILP_FlowBalance}) guarantees that all flows are integral for each binary variable assignment, so we can drop the integrality constraint \eqref{eq: ILP_VariableAllocation} and get the final MBLP formulation \eqref{eq: ILPNew}.

\subsection{Auxiliary Graph} \label{sec: Auxiliary Graph Construction}

We construct an auxiliary graph $\auxG = (\auxV, \auxE)$ as detailed below. Figure \ref{fig: AuxGraph} shows a pictorial depiction.

\begin{itemize}
    \item[(i)] \textit{Set of vertices \(\auxV = \cup_{\ell = 1}^{3}\auxV_\ell\)}. We define these sets below:
    \begin{itemize}
        \item \(\auxV_1 := \{(\verAux(r,t), \verAux^{\arr}(r,t),\verAux^{\dep}(r,t))| r\in \sector, t\in [\horizon]\}\): 
        
        We consider three replica for each vertiport \(r \in \sector\) at time \(t\in [\horizon]\), denoted as \(\verAux(r,t), \verAux^{\arr}(r,t),\) and \(\verAux^{\dep}(r,t)\). {These vertices, along with $\auxE_1$, $\auxE_2$, $\auxE_3$, and $\auxE_8$ defined later, embed capacity constraints and congestion costs into the graph structure.}
        
        \item \(\auxV_2 := \{\verAux({i,j,\tau})| i\in \OptSet, j\in \OptsUAVNum_i, \tau\in  T^{\dep}_{i,j}\}\): 
        
        For each \(i\in\OptSet, j\in \OptsUAVNum_i,\) we consider one vertex corresponding to all routes that have the same departure time. More formally, for every \(i\in\OptSet, j\in \OptsUAVNum_i,\) define 
        \(
        T^{\dep}_{i,j} := \cup_{k\in \UAVMenuNum_{i,j}}\left\{\depart_{i,j,k}\right\},\)
        to be the  set of unique departure times amongst all routes. 
        We consider one vertex corresponding to each \(i\in \OptSet, j\in \OptsUAVNum_i,\)  and \(\tau\in  T^{\dep}_{i,j}\), denoted as \(\verAux({i,j,\tau})\), {which, along with $\auxE_4$, $\auxE_5$, $\auxE_7$, and $\auxE_9$ defined later, embeds the route choice of the aircraft.}
        
        \item \(\auxV_3  := \{\verAux^{\sps}, \verAux^{\dps}\}\): 
        
        \(\verAux^{\sps}\) and \(\verAux^{\dps}\) denote the source and sink in the flow network {(to be described shortly)}. {These vertices, along with $\auxE_6$, $\auxE_7$, and $\auxE_8$, ensure flow conservation of the parking aircraft.}
    \end{itemize}    
    
    \item[(ii)] \textit{Set of edges \(\auxE = \cup_{\ell = 1}^{9}\auxE_\ell \subseteq \auxV\times \auxV\times \mathbb{Z}_+\times \mathbb{Z}_+\times \mathbb{R}\)}, where each edge is identified with a tuple \((r,r',\overline{\textbf{c}},\underline{\textbf{c}},\bar{\textbf{w}})\) such that \((i)\) \(r,r'\in \sector\) are the upstream and downstream vertiport on an edge, respectively, \((ii)\) \(\overline{\textbf{c}},\underline{\textbf{c}} \in \mathbb{Z}_+\) are the upper and lower bound on the capacity of the edge, respectively, and \((iii)\) \(\bar{\textbf{w}}\in \mathbb{R}\) is the edge weight.
    \begin{itemize}
    \item \(\auxE_1 := \{(\verAux^{\arr}(r,t), \verAux(r,t),\overline{\textbf{c}}=\arr(r,t), \underline{\textbf{c}}=0,\bar{\textbf{w}}=0)| r\in \sector, t\in [\horizon]\}\).
    \item \(\auxE_2 := \{(\verAux(r,t), \verAux^{\dep}(r,t),\overline{\textbf{c}}=\dep(r,t), \underline{\textbf{c}}=0,\bar{\textbf{w}}=0 )| r\in \sector, t\in [\horizon]\}\).
    \item \(\auxE_3 := \cup_{r\in \sector, t\in [\horizon-1]}\auxE_{3,r,t}\): 
    
    For every \(r\in \sector, t\in [H-1]\), we consider \(\capacity(r, t)\) edges connecting  \(\verAux(r,t)\) and \(\verAux(r,t+1)\). We denote this set by \(\auxE_{3,r,t}\). For any \(q\in [\capacity(r,t)]\), we denote the weight of the \(q-\)th edge in \(\auxE_{3,r,t}\) by \(\bar{\textbf{w}}_{q,r,t}\), and upper and lower capacity by \(\overline{\textbf{c}}_{q,r,t}\) and \( \underline{\textbf{c}}_{q,r,t}\), respectively. For any \(r\in \sector, q\in [\capacity(r,t)]\), \(\bar{\textbf{w}}_{q,r,t} = -\lambda(\cpc_{r,t}(q) - \cpc_{r,t}(q-1))\), \(\overline{\textbf{c}}_{q,r,t}= 1\), and \(\underline{\textbf{c}}_{q,r,t} = 0\). 
    \item \(\auxE_4 := \{(\verAux^\dep(\origsector_{i,j},\tau), \verAux({i,j,\tau}), \overline{\textbf{c}} = \underline{\textbf{c}} = \delta_{i,j,\tau}, \bar{\textbf{w}} = 0)|i\in\OptSet, j\in \OptsUAVNum_i, \tau\in T_{i,j}^{\dep}\backslash\{0\}\}\):
    
    \(\delta_{i,j,\tau}\in \{0,1\}\) is a variable defined later.
    \item \(\auxE_5 := \{(\verAux({i,j,\depart_{i,j,k}}), \verAux^\arr({\finalsector_{i,j,k}}, \arrive_{i,j,k}), \overline{\textbf{c}} = 1, \underline{\textbf{c}} = 0, \bar{\textbf{w}} = \rho_i b_{i,j,k})|i\in\OptSet, j\in \OptsUAVNum_i,k\in \UAVMenuNum_{i,j}\backslash\{\varnothing\}\}\).
    \item \(\auxE_6 := \{(\verAux^{\sps},\verAux(r,1), \overline{\textbf{c}} =\underline{\textbf{c}}= \bar{S}(r)-\sum_{i\in \OptSet}\sum_{j\in\OptsUAVNum_i}\delta_{i,j,0}\mathbf{1}(\origsector_{i,j}=r), \bar{\textbf{w}} =0)|r\in \sector\}\)\footnote{Recall that \(\bar{S}(r)\) is the state of occupancy of vertiport \(r\) at \(t=1\).}.
    
    \item \(\auxE_7 := \{(\verAux^{\sps},\verAux(i,j,0), \overline{\textbf{c}} = \underline{\textbf{c}} = \delta_{i,j,0}, \bar{\textbf{w}}= \rho_i b_{i,j,\varnothing})|i\in \OptSet, j\in \OptsUAVNum_i\}\):
    
    \(\delta_{i,j,0}\in \{0,1\}\) is a variable which would be defined shortly, and  {\(b_{i,j, \varnothing}\) is the bid placed by aircraft \(a_{i,j}\) on staying parked at the same location.}
    \item \(\auxE_8 := \cup_{r\in R}\auxE_{8,r}\): 
    
    For every \(r\in \sector\), we consider \(\capacity(r, H)\) edges connecting \(\verAux(r,H)\) and \(\verAux^{\dps}\). We denote these edges by \(\auxE_{8,r}\). 
    For any \(q\in [\capacity(r,H)]\), we denote the weight of the \(q-\)th edge in \(\auxE_{8,r}\) by \(\bar{\textbf{w}}_{q,r,H}\), and upper and lower capacity by \(\overline{\textbf{c}}_{q,r,H}\) and \(\underline{\textbf{c}}_{q,r,H}\) respectively.  For any \(r\in \sector, q\in [\capacity(r,H)]\), \(\bar{\textbf{w}}_{q,r,H} = -\lambda(\cpc_{r,H}(q) - \cpc_{r,H}(q-1)), \overline{\textbf{c}}_{q,r,H} = 1\), and \(\underline{\textbf{c}}_{q,r,H}= 0\).
    \item \(\auxE_{9} := \{(\verAux(i,j,0),\verAux(\origsector_{i,j},1), \overline{\textbf{c}} = \underline{\textbf{c}} = \delta_{i,j,0}, \bar{\textbf{w}}= 0)|i\in\OptSet, j\in \OptsUAVNum_i\}\).
    \end{itemize}
 
\end{itemize}

{
\begin{remark}
    In the preceding construction, the capacity of any outgoing edge (resp. incoming edge) from a node which does not have an incoming edge (resp. outgoing edge), other than \(\verAux^\sps\) and \(\verAux^\dps\), is set to 0. 
\end{remark}}

\subsection{Mixed Binary Linear Program Formulation}
\label{ssec: MILP}
We concatenate the weight, upper capacity bound, and lower capacity bound of each edge as $\overline{\mathbf{W}} \in \mathbb{R}^{|\auxE|}$, $\mathbf{\overline{C}} \in \mathbb{Z}_+^{|\auxE|}$, and $\mathbf{\underline{C}} \in \mathbb{Z}_+^{|\auxE|}$, respectively. Define an incidence matrix of the graph $\auxG$ as $\incidence \in \{-1, 0, 1\}^{|\auxV| \times |\auxE|}$, where 
\begin{equation}
\incidence_{ij} = 
\begin{cases}
        1, & \text{if edge $j$ ends at vertex $i$,} \\ 
        -1, & \text{if edge $j$ starts from vertex $i$,} \\ 
        0, & \text{otherwise}.
    \end{cases}
\end{equation}

Defining a truncated incidence matrix $\incidence_{\star}$ obtained from $\incidence$ by removing rows corresponding to \(\verAux^\sps\) and \(\verAux^\dps\), we have the following optimization problem.

\begin{subequations}
\label{eq: ILP}
\begin{align}
&\max_{\substack{\mathbf{A}, \when}} ~ \overline{\mathbf{W}}^\top\mathbf{A}\label{eq: ILP_Obj}\\
&\ \text{s.t. } \ \ \incidence_\star \mathbf{A} = \mathbf{0} \label{eq: ILP_FlowBalance} \\
& \ \ \ \ \ \ \ \mathbf{\underline{C}}(\delta) \leq \mathbf{A} \leq \mathbf{\overline{C}}(\delta)\label{eq: ILP_CapacityConstraint} \\
& \ \ \ \ \ \ \ \sum_{\tau\in T_{i,j}^{\dep}} \when_{i, j, \tau} = 1, \forall~ i \in \OptSet, j \in \OptsUAVNum_i \label{eq: ILP_BinaryVariable}\\
& \ \ \ \ \ \ \ \when_{i, j, \tau} \in \{0, 1\}, \forall~ i \in \OptSet, j \in \OptsUAVNum_i , \tau\in T_{i,j}^{\dep}\label{eq: ILP_VariableDelta} \\
& \ \ \ \ \ \ \ \mathbf{A} \in \mathbb{Z}^{|\auxE|}_+\label{eq: ILP_VariableAllocation}\\
& \ \ \ \ \ \ \ \mathbf{A}_{q+1,r,t} \! \leq \! \mathbf{A}_{q,r,t}, \forall r\in \sector, t\in [\horizon], q\in [\capacity(r,t)-1] \label{eq: ILP_Increasingq_t}.
\end{align}
\end{subequations}
Here, \eqref{eq: ILP_FlowBalance} denotes the ``flow balance'' constraint at every node in \(\auxV\backslash\auxV_3\); \eqref{eq: ILP_CapacityConstraint} denotes the capacity constraints where we have explicitly denoted the dependence of constraints on \(\delta\) (cf. definitions of \(\auxE_4\) and \(\auxE_7\)); \eqref{eq: ILP_BinaryVariable} and \eqref{eq: ILP_VariableDelta} denote the constraint that each aircraft must be allocated exactly one route; \eqref{eq: ILP_VariableAllocation} denotes the integrality constriants; \eqref{eq: ILP_Increasingq_t} denotes additional constraints which require that edges in \(\auxE_{3,r,t}\) and \(\auxE_{8,r}\) are allocated in an increasing order. 

Next, we highlight the connection between the optimization problems  \eqref{eq: AllocationMultipleCapacity} and \eqref{eq: ILP}.
    
\begin{lemma}\label{lem: EquivalenceOfFeasibleSet}
    Given the values of \(\textbf{A}_e\) for \(e\in \auxE_3\cup\auxE_5\cup\auxE_8\) that satisfy the capacity constraints \eqref{eq: ILP_CapacityConstraint}, there exists a unique feasible solution \((\textbf{A},\delta)\) that satisfies \eqref{eq: ILP_FlowBalance}-\eqref{eq: ILP_Increasingq_t}. 
\end{lemma}

\begin{proof}
See Appendix \ref{app: proof 1}.
\end{proof}

\begin{proposition}\label{prop: Connection}
    Suppose \((A^\dagger, \delta^\dagger)\) is an optimal solution to \eqref{eq: ILP}. Then \(\overline{\textbf{W}}^\top \textbf{A}^\dagger =  \max_{x\in X} \SW(x; B)\). 
    Additionally, using \(A^\dagger\) we can uniquely determine \(x^\dagger\in X\) such that \(x^\dagger\in \underset{x\in X}{\arg\max}~ \SW(x; B)\). 
\end{proposition}
\begin{proof}First, we show that, for every \(x\in X\), there exists a unique \((\textbf{A}(x), \delta(x))\) satisfying \eqref{eq: ILP_FlowBalance}-\eqref{eq: ILP_Increasingq_t} and \(\overline{\textbf{W}}^\top\textbf{A}(x) = \SW(x;B).\) Indeed, we construct \((\textbf{A}(x),\delta(x))\) such that
\begin{itemize}
        \item[(i)] for every \(e\!\in \!\auxE_{5}\), where   \((\verAux({i,j,\depart_{i,j,k}}) , \verAux^\arr({\finalsector_{i,j,k},\arrive_{i,j,k}}))\in e\) for some \(i\in \OptSet, j\in \OptsUAVNum_i, k\in \UAVMenuNum_{i,j}\), it holds that \(\textbf{A}_e(x) = x_{i,j,k}\); 
        \item[(ii)] for every \(r\in \sector, t\in [\horizon]\), and \(q\in [\capacity(r,t)]\), it holds that \(\textbf{A}_e(x) = \mathbf{1}(q\leq S(r,t,x))\), where \(e\) is the \(q-\)th  edge in \(\auxE_{3,r,t}\cup\auxE_{8,r}.\)
    \end{itemize}
The above construction specifies the values of \(\textbf{A}_e(x)\) for \(e\in \auxE_3\cup\auxE_5\cup\auxE_8\). Additionally, by Lemma \ref{lem: EquivalenceOfFeasibleSet}, there exists a unique feasible solution \((\textbf{A}(x),\delta(x))\), and we get
\begin{align*}
&\overline{\mathbf{W}}^\top\mathbf{A}(x) = \sum_{e\in \auxE} \bar{\textbf{w}}_e \mathbf{A}_e(x) = \sum_{e\in \auxE_{3}\cup \auxE_{5} \cup \auxE_{7} \cup \auxE_{8}}\bar{\textbf{w}}_e \mathbf{A}_e(x),
\end{align*}
where the last equality holds because \(\bar{\textbf{w}}_e=0\) for \(e\in \auxE_{1}\cup \auxE_{2} \cup \auxE_{4} \cup \auxE_{6} \cup \auxE_{9}.\)

Then, we examine each term. First, observe the following.
\begin{align*}
    \sum_{e\in \auxE_5} \bar{\textbf{w}}_e \mathbf{A}_e(x) &= \sum_{i\in \OptSet}\sum_{j\in \OptsUAVNum_i}\sum_{k\in \UAVMenuNum_{i,j}} \rho_i b_{i,j,k}x_{i,j,k}.
\\
    \sum_{e\in \auxE_7} \bar{\textbf{w}}_e \mathbf{A}_e(x) &= \sum_{i\in \OptSet}\sum_{j\in \OptsUAVNum_i} \rho_i b_{i,j,\phi}x_{i,j,0}.
\end{align*}
Next, we use the definition of weights in \(\auxE_{3,r,t}\).
\begin{align*}
    &\sum_{e\in \auxE_3} \bar{\textbf{w}}_e \mathbf{A}_e(x)  =\sum_{r\in \sector} \sum_{t=1}^{\horizon-1} \sum_{e\in \auxE_{3,r,t}} \bar{\textbf{w}}_{e} \mathbf{A}_{e}(x) \\&=\sum_{r\in \sector} \sum_{t=1}^{\horizon-1} \sum_{q=1}^{\capacity(r,t)} \bar{\textbf{w}}_{q,r,t} \mathbf{A}_{q,r,t}(x)\\&=-\lambda\sum_{r\in \sector} \sum_{t=1}^{\horizon-1} \sum_{q=1}^{S(r,t,x)} \left(\cpc_r(q)-\cpc_r(q-1)\right) \\&= -\lambda\sum_{r\in \sector} \sum_{t=1}^{\horizon-1} \cpc_r(S(r,t,x)).
\end{align*}
Similarly, we get $\sum_{e\in \auxE_8} \bar{\textbf{w}}_e \mathbf{A}_e(x) = -\lambda\sum_{r\in \sector} \cpc_r(S(r,H,x))$.

To summarize, we obtain   
\begin{align}\label{eq: xToA}
\overline{\textbf{W}}^\top\textbf{A}(x) = \SW(x;B). 
\end{align}
Using this, we conclude that  
\begin{equation}\label{eq: Forward}
\begin{aligned}
    &\max_{x\in X}\SW(x;B) = \max_{x\in X}\overline{\textbf{W}}^\top\textbf{A}(x) \\& \leq  \max_{(\textbf{A},\delta) ~\text{s.t.}~ \eqref{eq: ILP_FlowBalance}-\eqref{eq: ILP_Increasingq_t}}\overline{\textbf{W}}^\top\textbf{A}= \overline{\textbf{W}}^\top\textbf{A}^\dagger.
\end{aligned}
\end{equation}

Next, we show that for every \((\textbf{A}, \delta)\) satisfying \eqref{eq: ILP_FlowBalance}-\eqref{eq: ILP_Increasingq_t}, there exists \(x(\textbf{A},\delta)\in X\) such that \(\SW(x(\textbf{A},\delta)) = \overline{\textbf{W}}^\top \textbf{A}\). Indeed, we construct \(x(\textbf{A},\delta)\) such that for every \(i\in \OptSet, j\in \OptsUAVNum_i, k\in \UAVMenuNum_{i,j}\) it holds that \(x_{i,j,k} = \textbf{A}_e\) for \(e\in \auxE_5\) such that \((\verAux({i,j,\depart_{i,j,k}}), \verAux^\arr({\finalsector_{i,j,k}}, \arrive_{i,j,k}))\in e\) or \(e\in \auxE_9\) such that \((\verAux({i,j,0}), \verAux({\finalsector_{i,j,k}}, 1)) \in e\). Note that due to capacity constraints on these edges, \(x_{i,j,k}\in\{0,1\}\). Additionally, the flow balance at the nodes of the form \(\verAux(i,j,\tau)\), for some \(i\in \OptSet, j\in\OptsUAVNum_i, \tau\in T_{i,j}^\dep\),  ensures that 
\begin{align*}
&\delta_{i,j,\tau} = \!\!\!\!\sum\limits_{k\in \UAVMenuNum_{i,j}}\sum\limits_{e\in \auxE_9}\!\!\!\textbf{A}_{e}\mathbf{1}((\verAux({i,j,0}) , \verAux({\finalsector_{i,j,k},1}))\!\in\! e, \tau= 0)\\
&+\!\!\!\!\sum\limits_{k\in \UAVMenuNum_{i,j}}\sum\limits_{e\in \auxE_5}\!\!\!\!\textbf{A}_{e}\mathbf{1}((\verAux({i,j,\depart_{i,j,k}}) , \verAux^\arr({\finalsector_{i,j,k},\arrive_{i,j,k}}))\!\in\! e, \tau\!=\!\!\depart_{i,j,k}).
\end{align*}
Summing over $\tau$, we get   
\begin{align*}
    &\sum_{\tau\in T_{i,j}^\dep}\!\!\delta_{i,j,\tau} =\!\! \sum_{k\in \UAVMenuNum_{i,j}}\sum_{e\in \auxE_9}\textbf{A}_{e}\mathbf{1}((\verAux({i,j,0}) , \verAux({\finalsector_{i,j,k}, 1}))\in e)\\
    &\quad + \sum_{k\in \UAVMenuNum_{i,j}}\sum_{e\in \auxE_5}\textbf{A}_{e}\mathbf{1}((\verAux({i,j,\depart_{i,j,k}}) , \verAux^\arr({\finalsector_{i,j,k},\arrive_{i,j,k}}))\in e) \\
    &\quad \quad\quad \quad\ = \sum_{k\in\UAVMenuNum_{i,j}}x_{i,j,k}(\textbf{A},\delta).
\end{align*}
Using \eqref{eq: ILP_BinaryVariable}, we conclude that \(\sum_{k\in \UAVMenuNum_{i,j}} x_{i,j,k}(\textbf{A},\delta)= 1\). 

Next, we use  the flow balance at nodes of the form \(\verAux^\arr(r,t)\), for every \(r\in \sector, t\in [\horizon]\), to ensure that 
\begin{align*}
\sum\limits_{\substack{i\in\OptSet, j\in\OptsUAVNum_i,\\ k\in \UAVMenuNum_{i,j},e\in \auxE_5}} \!\!\!\!\!\!\!\!\!\!\!\!\textbf{A}_{e}\mathbf{1}((\verAux({i,j,\depart_{i,j,k}}) , \verAux^\arr({\finalsector_{i,j,k}, \arrive_{i,j,k}}))\!\!\in\!\! e, \finalsector_{i, j, k} \!= \!r, \arrive_{i, j, k} \!=\! t) \leq \arr(r, t).
\end{align*}
By \(\sum_{e\in \auxE_5}\textbf{A}_e\mathbf{1}((\verAux({i,j,\depart_{i,j,k}}) , \verAux^\arr({\finalsector_{i,j,k}, \arrive_{i, j, k}}))\!\!\in\!\! e) \!\!=\!\! x_{i,j,k}(\textbf{A},\delta)\), we get \(\sum_{i \in \OptSet } \sum_{j \in \OptsUAVNum_i} \sum_{k \in \UAVMenuNum_{i, j}} \allocation_{i, j, k} \mathbf{1}(\finalsector_{i, j, k} = r, \arrive_{i, j, k} = t) \!\!\leq\!\! \ArrivalCap(r,t)\)\footnote{When $t=1$, the arrival capacity constraints are trivially satisfied since there is no incoming aircraft.}. Analogously, the flow balance equations at the nodes of the form \(\verAux^\dep(r,t)\), for some \(r\in \sector, t\in [\horizon]\), ensure that \(\sum_{i\in\OptSet}\sum_{j\in \OptsUAVNum_i} \sum_{k \in \UAVMenuNum_{i, j}} \allocation_{i, j, k} \mathbf{1}(\origsector_{i, j} = r, \depart_{i, j, k} = t) \leq \DepartureCap(r, t).\) 
Finally, we can establish \(S(r,t,x(\textbf{A},\delta)) = \sum_{q=1}^{\capacity(r,t)}\textbf{A}_{q,r,t} \) through the flow balance equation at \(\verAux(r,t)\) and \eqref{eq: StateEvolution}. Since \(\sum_{q=1}^{\capacity(r,t)}\textbf{A}_{q,r,t}\leq \capacity(r,t)\), due to the capacity constraints on the edge \(\auxE_{3,r,t}\), it holds that   \(S(r,t,x(\textbf{A},\delta)) \leq \capacity(r,t)\). Thus, we conclude that \(x(\textbf{A},\delta)\in X\). 

Additionally, using the analysis to show \eqref{eq: xToA} in the backward direction and the construction of \(x(\textbf{A},\delta)\), we can establish that \(\SW(x(\textbf{A},\delta))= \overline{W}^\top \textbf{A}\). Thus, 
we conclude that  
\begin{equation}
\begin{aligned}\label{eq: reverse}
&\overline{\textbf{W}}^\top\textbf{A}^\dagger=  \max_{(\textbf{A},\delta) ~\text{s.t.}~ \eqref{eq: ILP_FlowBalance}-\eqref{eq: ILP_Increasingq_t}}\overline{\textbf{W}}^\top\textbf{A} \\&= \max_{(\textbf{A},\delta) ~\text{s.t.}~ \eqref{eq: ILP_FlowBalance}-\eqref{eq: ILP_Increasingq_t}}\SW(x(\textbf{A},\delta)) \leq \max_{x\in X} \SW(x). 
\end{aligned}
\end{equation}

By \eqref{eq: Forward} and \eqref{eq: reverse}, we get \(\overline{\textbf{W}}^\top\textbf{A}^\dagger=  \max_{x\in X} \SW(x)\).
\end{proof}

\subsection{Reduction to Mixed Binary Linear Program}
\label{ssec: MBLP}
Instead of solving \eqref{eq: ILP}, we can obtain \((\textbf{A}^\dagger, \delta^\dagger)\) by solving the following MBLP. We establish this fact in Proposition \ref{prop: MBLP}.
\begin{subequations}
\label{eq: ILPNew}
\begin{align}
&\max_{\substack{\mathbf{A}, \when}} ~ \overline{\mathbf{W}}^\top\mathbf{A}\label{eq: ILP_Obj2}\\
&\ \text{s.t. } \eqref{eq: ILP_FlowBalance}-\eqref{eq: ILP_VariableDelta} \\
& \ \ \ \ \ \ \ \mathbf{A} \in \mathbb{R}^{|\auxE|}_+\label{eq: ILP_VariableAllocation2}.
\end{align}
\end{subequations}

\begin{proposition}\label{prop: MBLP}
    The optimal values of \eqref{eq: ILP} and \eqref{eq: ILPNew} are equal.
\end{proposition}
\begin{proof}
First, we prove that we can drop \eqref{eq: ILP_Increasingq_t} when solving \eqref{eq: ILP}. Suppose there exists  \(r\in \sector, t\in [\horizon]\), \(q\in [\capacity(r,t)-1]\) such that  \(\textbf{A}^\dagger_{q,r,t}  < \textbf{A}^\dagger_{q+1,r,t}\). By swapping the value of \(\textbf{A}^\dagger_{q+1,r,t}\) with that of \(\textbf{A}^\dagger_{q,r,t}\), we get a new feasible allocation with a weakly higher objective value. This is because \( \bar{\textbf{w}}_{q+1,r,t} \leq \bar{\textbf{w}}_{q,r,t}\) as \(\lambda \geq 0\) and \(\cpc_{r, t}(\cdot)\) is discrete convex. 
Then, for any feasible value of \(\delta\), the optimization problem \eqref{eq: ILP} is an integer linear program where the constraint matrix $\incidence_\star$ satisfies total unimodularity, so it is guaranteed to have an integral solution \cite[Chapter 19]{10.5555/17634}. 
\end{proof}

{For any fixed values of binary variables \((\delta_{i,j,\tau})_{i\in\OptSet, j\in \OptsUAVNum_i, \tau\in T_{i,j}^\dep}\), the optimization problem \eqref{eq: ILP} is a maximum-weight flow problem. Thus, one can enumerate all the departure time combinations, and solve each maximum-weight flow problem with the number of scenarios being $\prod_{i\in\OptSet, j\in\OptsUAVNum_i} |\{\depart_{i, j, k}| k \in \UAVMenuNum_{i, j}\}|$. The complete problem can be solved efficiency using the above MBLP approach, which will provide speed-up due to some techniques implemented in commercial solvers such as branch and bound, cutting-plane methods, etc.}

\section{Discussions}
We show how the proposed mechanism generalizes existing works in Subsection \ref{ssec: UnitCapacity} and present some extensions in Subsection \ref{Sec: speed-up}.

\subsection{Connections to Existing Mechanisms}
\label{ssec: UnitCapacity}
We consider \(H=1\), $\ArrivalCap(r, 1) = \infty$, $\DepartureCap(r, 1) = \infty$, $\forall r \in \sector$, and $|A_i| = 1$, $\forall i \in \OptSet$.
\begin{itemize}
    \item[(i)] \textbf{Air Traffic Protocol:} When we treat each vertiport $r \in \sector$ as a sector with $\capacity(r, 1)$ being the sector capacity, our model generalizes the problem studied in \cite{balakrishnan2022cost}, where the authors did not consider arrival and departure capacities and assumed single-aircraft FOs.
    \item[(ii)] \textbf{Airport Time Slot Auction:} When we treat each vertiport $r \in \sector$ as a time slot with $\capacity(r, 1)$ being the slot capacity, our model subsumes the framework in \cite{DIXIT2023102971}. Therefore, our formulation becomes a two-sided matching problem as detailed in \cite{DIXIT2023102971} and is subject to a faster strongly polynomial-time algorithm.
\end{itemize}

\subsection{Extensions of the Proposed Mechanism}
\label{Sec: speed-up}

\begin{itemize}
    \item[(i)] {\textbf{Arrival, Departure, and Airborne Congestion:} To consider congestion due to arriving and departing aircraft, we can apply the same technique in $\auxE_3$ to $\auxE_1$ and $\auxE_2$ by constructing corresponding edge weights. To consider airborne congestion, we treat waypoints in the airspace as vertiports and setting corresponding capacities and congestion costs.}
    \item[(ii)] \textbf{External Demand:} Aircraft that are not available in the service area of the SP at \(t=0\) can be incorporated in our framework by setting $\origsector_{i, j} = \outside$ and $\depart_{i, j, k} = 0, \forall k \in \UAVMenuNum_{i, j}$\footnote{In this case, $\origsector_{i, j}$ and $\depart_{i, j, k}$ do not affect our analysis, so we can set them arbitrarily.}. 
    \item[(iii)] \textbf{Entire Trajectory:} We can extend each route to an entire trajectory with multiple vertiport-time pairs. By setting a binary variable for each route and combining those variables when two routes only differ in one time slot, we can apply the same MBLP approach.
    \item[(iv)] \textbf{Cancellation Policy:} It is possible to cancel or re-allocate some of the previously scheduled flights due to changing vertiport capacities or newly emerging aircraft. While there is no single re-allocation policy, it is typical to consider three aspects: congestion, efficiency, and fairness, where we cancel flights from congested vertiports, with low valuations, or at random, respectively.
\end{itemize}

%% file: CDC_Sections/Conclusion2.tex
\section{Conclusion}
In this work, we propose an auction mechanism to incentivize fleet operators to report their valuations truthfully and consequently perform a socially optimal allocation of vertiport access. This approach adapts the popular Vickrey–Clarke–Groves mechanism while considering the egalitarian, congestion-aware, and computational issues. The proposed framework could be of interest beyond air traffic management, such as multi-robot coordination. Code associated with this paper can be found at \url{https://github.com/victoria-tuck/IC-vertiport-reservation}; in a follow-up paper, we shall provide a numerical analysis of the mechanism's performance.
Several intriguing avenues exist for future research. First, we would like to extend the auction mechanism to include waypoints in airspace, thus moving toward a more complete air traffic flow management formulation. Second, {a careful analysis of the effect of flight operator weights in our proportional fairness metric is needed.} Also, different fairness notions have been considered in airspace and other areas, such as reversals, takeovers, and priority guarantees \cite{doi:10.1287/trsc.2014.0567, 9965945}. It is interesting to compare different formulations, both theoretically and empirically.

%% file: CDC_Sections/Appendix2.tex
\section{Proof of Lemma \ref{lem: EquivalenceOfFeasibleSet}}
\label{app: proof 1}
\begin{proof}
    First, note that any feasible solution to \eqref{eq: ILP_FlowBalance}-\eqref{eq: ILP_Increasingq_t} has the same value of \(\textbf{A}_e(x)\) for \(e\in \auxE_6\cup\auxE_7\cup\auxE_9\) since the lower and upper bound on capacity are the same on these edges by construction. Thus, it is sufficient to show that the values of  \(\textbf{A}_e\) for \(e\in \auxE_3\cup\auxE_5\cup\auxE_6\cup\auxE_7\cup\auxE_8\cup\auxE_9\) uniquely determine a feasible solution \((\textbf{A},\delta)\) that satisfies \eqref{eq: ILP_FlowBalance}-\eqref{eq: ILP_Increasingq_t}. Particularly, we will show that we can uniquely recover the values of \(\textbf{A}_e\) for \(e\in \auxE_1\cup\auxE_2\cup\auxE_4.\)
    
    To show this claim, we  leverage the flow balance constraint \eqref{eq: ILP_FlowBalance} at every node. Below, we state the incoming and outgoing edges from every type of node in the network. 
    \begin{center}
\begin{tabular}{||c || c | c||} 
 \hline
 Vertex & Incoming Edges & Outgoing Edges \\ [0.5ex] 
 \hline\hline
 \(\verAux^\arr(r,t)\) & \({\auxE_5}\) & \(\auxE_1\) \\ 
 \hline
 \(\verAux(i,j,\tau)\) &  \(\auxE_4\)  & \({\auxE_5}\)\\
 \hline
  \(\verAux^\dep(r,t)\) &  \(\auxE_2\)  & \({\auxE_4}\) \\
 \hline
 \(\verAux(i,j,0)\) & \(\auxE_7\) & \(\auxE_9\) \\
 \hline
 \(\verAux(r,t)\) & \(\auxE_{1}, \auxE_{3}, \auxE_{6}, \auxE_{9}\) &  \(\auxE_{2}, \auxE_{3}, \auxE_{8}\) \\ 
 \hline
\end{tabular}
\end{center}
Note that flow balance at nodes of the form \(\verAux^\arr(r,t)\) will determine the values \(\textbf{A}_e\) on edge \(\auxE_1\), as we know these values for edges \(\auxE_5\). Next, flow balance at  nodes of the form \(\verAux(i,j,\tau)\) will determine the values \(\textbf{A}_e\) on edge \(\auxE_4\), as we know these values for edges \(\auxE_5\). This and the capacity constraints on \(\auxE_4\), ensure that we know the value of \(\delta\). Next, flow balance at nodes of the form \(\verAux^\dep(r,t)\) will determine the values \(\textbf{A}_e\) on edge \(\auxE_2\), as we can uniquely determine these values on \(\auxE_4\). Finally, flow balance at nodes of the form \(\verAux(r,t)\) will determine the values \(\textbf{A}_e\) on edge \(\auxE_2\), as we can uniquely determine these values on \(\auxE_1\cup\auxE_3\cup\auxE_6\cup \auxE_8\cup\auxE_9\).

\end{proof}

%% file: CDC_Sections/TableOfNotations.tex
\begin{center}
\begin{tabular}{||c l||} 
 \hline
Notation & Description  \\ [0.5ex] 
 \hline\hline
 \(\sector\) & Set of vertiports
 \\
 \hline
 \(\OptSet\) & Set of fleet operators\\ 
 \hline
 \(\OptsUAVNum_i\) & Set of eVTOL aircraft in the fleet of operator \(i\in \OptSet\)   \\
 \hline
 \(\horizon\) & Scheduling horizon \\
 \hline
 \(\UAV_{i,j}\) &  The identification of the \(j-\)th aircraft in \(\OptsUAVNum_i\) \\
 \hline
 \(\origsector_{i,j}\) & Origin vertiport of aircraft \(a_{i,j}\) 
 \\ 
   \hline
 \(\UAVMenuNum_{i,j}\) & Set of available routes of aircraft \(a_{i,j}\)\\
  \hline
 \(\depart_{i,j,k}\) & Departure time for aircraft \(a_{i,j}\) if it chooses the \(k-\)th route in \(\UAVMenuNum_{i,j}\)
 \\ 
 \hline
 \(\finalsector_{i,j,k}\) & Destination vertiport of aircraft \(a_{i,j}\) if it chooses the \(k-\)th route in \(\UAVMenuNum_{i,j}\)
 \\ 
 \hline 
\(\arrive_{i,j,k}\) & Arrival time of aircraft \(a_{i,j}\) at \(\finalsector_{i,j,k}\) if it chooses the \(k-\)th route in \(\UAVMenuNum_{i,j}\)
 \\ 
 \hline 
\(\val_{i,j,k}\) & Valuation derived by aircraft \(a_{i,j}\) if it is allocated the \(k-\)th route in \(\UAVMenuNum_{i,j}\)
 \\ 
  \hline 
\(\bid_{i,j,k}\) &  Bid for aircraft \(a_{i,j}\) to be allocated the the \(k-\)th route in \(\UAVMenuNum_{i,j}\)
 \\ 
 \hline 
\(\allocation_{i,j,k}\) &  Binary variable denoting whether aircraft \(\UAV_{i,j}\) is allocated the \(k-\)th route in \(\UAVMenuNum_{i,j}\)
 \\
 \hline 
\(S(r,t,x)\) & Number of aircraft at vertiport \(r\in \sector\) at time \(t \in [\horizon]\) under allocation \(x\)
 \\
  \hline 
\(\arr(r,t)\) & Arrival capacity of vertiport \(r\in\sector\) at time \(t\in [\horizon]\)
 \\
 \hline
 \(\dep(r,t)\) & Departure capacity of vertiport \(r\in\sector\) at time \(t\in [\horizon]\)
 \\
  \hline
 \(\capacity(r,t)\) & Parking capacity of vertiport \(r\in\sector\) at time \(t\in [\horizon]\)
 \\
 \hline
 \(\SW(x;V)\) & Social welfare under allocation \(x\) if the valuation of aircraft is \(V=(v_{i,j,k})_{i\in F, j\in A_i, k\in \UAVMenuNum_{i,j}}\)
 \\
 \hline
 \(\pseudoBid\) & Pseudo-bids
 \\
 \hline
 \(C_{r,t}(\cdot)\) & Congestion function of vertiport \(r\in\sector\) at time \(t\in [H]\)
 \\
 \hline
 \(\auxG\) & Auxiliary graph for the optimization algorithm
 \\
 \hline
 \(\auxV\) & Set of vertices of the auxiliary graph
 \\
 \hline
 \(\auxE\) & Set of edges of the auxiliary graph
 \\
 [1ex] 
 \hline
\end{tabular}
\end{center}